\documentclass[12pt]{article}
\usepackage{amsmath,amsfonts, amsthm,amssymb,hyperref,bbm,eucal,verbatim,enumerate,leftidx}
\usepackage{placeins}

\usepackage{pstricks, pst-node}
\psset{arrows=->, labelsep=3pt, mnode=circle}
\input pst-plot
\numberwithin{equation}{section}

\newtheorem{thm}{Theorem}
\newtheorem{lemma}{Lemma}[section]

\newtheorem{cor}[lemma]{Corollary}

\newtheorem*{bel*}{Belief}

\newtheoremstyle{named}{}{}{\itshape}{}{\bfseries}{.}{.5em}{\thmnote{#3}}
\theoremstyle{named}
\newtheorem*{thm1*}{Theorem}

\def\Ai{\mathrm{Ai}}
\def\tr{\operatorname{tr}}

\newcounter{remnr1}
\def\res1{
    \addtocounter{remnr1}{1}
    \vspace{2mm}\noindent{\bf (\Alph{remnr1})} }

\def\ltau{\,^{(\tau)}\!}
\def\ltaup{\,^{(\tau')}\!}
\def\eqD{\overset{\mathrm{distr}}{=}}
\def\AD{\mathrm{A\!D}}

\def\invb{\overset{\mathrm{inv},\beta}}
\def\invo{\overset{\mathrm{inv},1}}
\def\invu{\overset{\mathrm{inv},2}}
\def\t{{\bf r}}
\def\T{\mathfrak{R}}

\begin{document}

\title{A limit theorem at the spectral edge for corners of time-dependent Wigner matrices}
\author{Sasha Sodin\footnote{Department of Mathematics, Princeton University, Princeton, NJ 08544, USA \& School of Mathematical Sciences, Tel Aviv University, Ramat Aviv, Tel Aviv 6997801,
Israel. E-mail: asodin@princeton.edu. 
Supported in part by NSF grant PHY-1305472.}}
\maketitle

\begin{abstract}
For the eigenvalues of principal submatrices of stochastically evolving 
Wigner matrices, we construct and study the edge scaling limit: a random 
decreasing sequence of continuous functions of two variables, which at 
every point has the distribution of the Airy point process. The analysis is
based on the methods developed by
Soshnikov to study the extreme eigenvalues of a single Wigner matrix. 
\end{abstract}

\section{Introduction}

In this paper, we consider a time-dependent infinite Hermitian random matrix
\begin{equation}
\ltau{H} = \left( \ltau{H}(i, j) \right)_{i,j \geq 1} \quad (-1 \leq \tau \leq 1)
\end{equation}
such that the entries
\[ \left(\ltau{H}(i,j) \right)_{i \leq j} \]
of $\ltau{H}$ above the diagonal are independent random processes, and the
following properties hold:
\begin{enumerate}
 \item[A1)] the distribution of $\ltau{H}(i, j)$ is symmetric: 
 $\ltau{H}(i, j) \eqD - \ltau{H}(i, j)$ ($i,j \geq 1$, jointly in $\tau$);
 \item[A2)] the tails of $\ltau{H}(i, j)$ are subgaussian: $\mathbb{E} |\ltau{H}(i,j)|^{2k} \leq (C_0k)^k$
($k \geq 1$);
 \item[A3)] the covariances of the off-diagonal elements near $\tau = 0$ satisfy either
\[\tag{A3$_1$}
\begin{split}\mathbb{E} ^{(\tau_1)}\!{H}(i, j) ^{(\tau_2)}\!{H}(i, j) &= 1 -|\tau_1 - \tau_2| + o(|\tau_1|+|\tau_2|) \\
\mathbb{E} ^{(\tau_1)}\!{H}(i, j) ^{(\tau_2)}\!{H}(j, i)&= 1 - |\tau_1 - \tau_2| + o(|\tau_1|+|\tau_2|)
\end{split}\quad (i \neq j)\]
or
\[\tag{A3$_2$}\begin{split}\mathbb{E} ^{(\tau_1)}\!{H}(i, j) ^{(\tau_2)}\!{H}(i, j) &= 0 \\
\mathbb{E} ^{(\tau_1)}\!{H}(i, j) ^{(\tau_2)}\!{H}(j, i)&= 1 - |\tau_1 - \tau_2| + o(|\tau_1|+|\tau_2|)
\end{split}\quad (i \neq j)~,\]
where the $o$-term is uniform in $i,j \geq 1$.
\end{enumerate}
As customary in random matrix theory, the case
A3$_1$) (corresponding to the universality class of
orthogonal symmetry, $\beta=1$) and A3$_2$) 
(corresponding to the universality class of unitary
symmetry, $\beta=2$) are considered in parallel.
The setting is similar to that considered by Borodin \cite{B2}, see below. 

\vspace{2mm}\noindent To the infinite matrix $\ltau{H}$ we associate a nested sequence of finite submatrices
$\left(\ltau{H}^{(N)}\right)_{N \geq 1}$,
\[ \ltau{H}^{(N)} = \left( \ltau{H}(i, j) \right)_{1 \leq i,j \leq N}~. \]
The eigenvalues
\[ \ltau{\xi}^{(N)} = \left\{ \ltau{\xi}_1^{(N)} \geq \ltau{\xi}_2^{(N)} \geq \cdots \geq \ltau{\xi}_N^{(N)} \right\}\]
boast the interlacing property
\[ \ltau{\xi}_1^{(N+1)} \geq \ltau{\xi}_1^{(N)} \geq \ltau{\xi}_2^{(N+1)} \geq \ltau{\xi}_2^{(N)} \geq \ltau{\xi}_3^{(N+1)} \geq \cdots~, \]
thus, for every $-1 \leq \tau \leq 1$, $\ltau{\Xi} = \left( \ltau{\xi}^{(N)}_j \right)_{1 \leq j \leq N < \infty}$ is a triangular
array, called the Wigner corner process.

\paragraph{Example 1: Dyson Brownian motion} The most important example of a time-dependent random matrix is arguably the Dyson Brownian motion,
introduced  by Dyson \cite{D} (see also the book of Mehta \cite[Chapter 9]{Mehta}). One of the versions
is as follows. Let $\ltau{X}$, $\ltau{Y}$ be infinite matrices composed of independent copies
of the Ornstein--Uhlenbeck process,
\[ \mathbb{E} \ltau{X}(i, j) \ltaup{X}(i, j) = \mathbb{E} \ltau{Y}(i, j) \ltaup{Y}(i, j)  =
e^{-|\tau - \tau'|}~.\]
For $\beta = 1$, set $\ltau{H} = \left( \ltau{X} + \ltau{X}^T \right) / \sqrt{2}$. At a fixed time
$\tau$, the matrix $\ltau{H}$ has the distribution of the Gaussian Orthogonal Ensemble (GOE), 
i.e.\ the probability density of the submatrix $\ltau{H}^{(N)}$ with respect to the Lebesgue
measure on the space of $N \times N$ real symmetric matrices is equal to
\begin{equation}\label{eq:GOE}
\frac{1}{Z_1^{(N)}} \exp \left\{ - \frac{1}{2} \tr (H^{(N)})^2 \right\}~.
\end{equation}
For $\beta = 2$, set $\ltau{H} = \left( \ltau{X} + \ltau{X}^T  + i \ltau{Y} -i \ltau{Y}^T \right) / 2$. At
a fixed time $\tau$,  the matrix $\ltau{H}$ has the distribution of the Gaussian Unitary Ensemble 
(GUE), i.e.\ the probability density of the submatrix $\ltau{H}^{(N)}$ with respect to the Lebesgue
measure on the space of $N \times N$ Hermitian matrices is equal to
\begin{equation}\label{eq:GUE}
\frac{1}{Z_2^{(N)}} \exp \left\{ -\tr (H^{(N)})^2 \right\}~. 
\end{equation}
We denote the invariant ensembles constructed above $\ltau\invb{H}$, and refer to  
the books \cite{Mehta, Fbook, PShch} for properties, history, and references.
\paragraph{Example 2: Resampled Wigner matrix} The second
example is a resampled Wig\-ner matrix, constructed
as follows. Pick a random Hermitian matrix
$H = (H(i, j))_{i,j\geq 1}$ such that the entries $(H(i,j))_{i\leq j}$ are independent, and satisfy
\begin{enumerate}
\item[a1)] $H(i, j) \eqD - H(i, j)$;
\item[a2)] $\mathbb{E} |H(i,j)|^{2k} \leq (C_0 k)^k$, $k \geq 1$;
\item[a3)] for $i \neq j$,  $\mathbb{E} |H(i, j)|^2 = 1$, and 
$\mathbb{E} H(i, j)^2 = \begin{cases} 1~, &\beta = 1 \qquad \text{a3}_1)\\ 0~, &\beta = 2 \qquad \text{a3}_2)\end{cases}$.
\end{enumerate}
The time-dependent random matrix $\ltau{H}$ is obtained by resampling every entry of $H$ at
independent Poisson times (of intensity one).

\paragraph{Edge scaling} The edge scaling of the time-dependent triangular array  $\ltau{\Xi}$ is carried out as follows. For $M \geq 1$ (which will play the r\^ole of a
large parameter, and which we assume, for simplicity of notation, to be integer),
define  ${_M}\tau(s)$ and ${_M}N(t)$ ($s,t \in \mathbb{R}$) by the formul{\ae}
\[\begin{split}
{_M}\tau(s) &= s M^{-1/3}~, \\
{_M}N(t) &= M(1 + 2 t M^{-1/3})~.
  \end{split}\]
If ${_M}N(t)$ is a natural number, set
\begin{equation}\label{eq:scaling}
{_M}\lambda_j(s, t) = M^{1/6} \left( ^{(_M\tau(s))}\xi_j^{({_M}N(t))} - 2 \sqrt{{_M}N(t)} \right)~;
\end{equation}
then interpolate ${_M}\lambda_j(s, \bullet)$ as a piecewise linear function, and extrapolate $_M\lambda_j(\bullet,\bullet)$
as a constant.

So far we have constructed, for every $M \geq 1$, a random decreasing sequence ${_M}\Lambda$ of functions ${_M}\lambda_j(s, t)$ ($s,t \in \mathbb{R}$); $s$ is the rescaled time, whereas $t$ is the rescaled number of the corner.

\vspace{2mm}\noindent
The first theorem asserts that the scaling limit exists, and depends  on the choice of the matrix $\ltau{H}$
only via $\beta \in \{1, 2\}$. Both the statement and the proof may be viewed as extensions
of Soshnikov's universality theorem \cite{Sosh}, and are firmly based on the arguments 
developed there.

\begin{thm}\label{thm:sosh}
For $\beta \in \{1,2\}$, there exists a (non-trivial) random process
\[ \AD_\beta(s, t) = (\lambda_1(s, t) \geq \lambda_2(s, t) \geq \cdots) \quad (s, t \in \mathbb{R}) \]
such that, for any time-dependent random matrix $\ltau{H}$ satisfying A1), A2), and A3$_\beta$),
one has:
\begin{equation}\label{eq:conv} {_M}\Lambda(s, t) \overset{\mathrm{distr}}{\to} \AD_\beta(s, t) \quad (M \to \infty)\end{equation}
in the sense of finite-dimensional marginals.
\end{thm}

In Section~\ref{s:prop}, we give a characterisation of the process $\AD_\beta(s, t)$ (which is not very compact
but in principle suited, for example, for the study of asymptotics). For
now, we summarise some  basic features.
\begin{thm}\label{thm:prop}
Fix $\beta \in \{1, 2\}$. The random process $\AD_\beta(s, t)$ boasts the following properties:
\begin{enumerate}
 \item for any $k \geq 1$, the joint distribution of $(\AD_\beta(s_p, t_p))_{1 \leq p\leq k}$ depends only
on the $\ell_1$ distances $|s_p - s_r| + |t_p - t_r|$ ($1 \leq p < r \leq k$);
 \item $\AD_\beta(s, t)$ has a modification in which every $\lambda_j(s, t)$ is a continuous 
function of $(s, t) \in \mathbb{R}^2$.
\end{enumerate}
\end{thm}

From the first statement of Theorem~\ref{thm:prop}, the distribution of the restriction $\AD_\beta(\t(a))$ of $\AD_\beta$
to an $\ell_1$ geodesic
\[ \t: \mathbb{R} \to \mathbb{R}^2~, \quad
\| \t(a) - \t(a_1) \|_1 = |a - a_1| \quad (a,a_1 \in \mathbb{R}) 
\]
does not depend on the choice of the geodesic. For $\beta=2$, this distribution is given by the Airy$_2$
time-dependent point process, see Ferrari \cite[Section~4.4]{Fer} and the discussion below.

\vspace{1mm}\noindent
In the next statement, we upgrade Theorem~\ref{thm:sosh} to a functional limit theorem. The
most general result would state that (\ref{eq:conv}) holds in an appropriate Skorokhod space of c\`adl\`ag functions (see Billingsley \cite{Bil}).
Such a statement could probably be proved using the methods of the current paper. We favour simplicity and prove the following result, pertaining to corners of one random matrix.
\begin{thm}\label{thm:conv}
Let $\beta = 1$ or $2$, and let ${H}$ be an infinite random matrix satisfying
a1), a2), and a3$_\beta$), such that moreover
\begin{equation}\label{eq:unimod}
|{H}_{i,j}| = 1 \quad (i \neq j)~, \quad {H}(i, i) = 0 \quad (-1 \leq \tau \leq 1, \, 1 \leq i, j)~.
\end{equation}
Then
\[ {_M}\Lambda(t) \overset{\mathrm{distr}}{\to} \AD_\beta(0, t)  \quad (M \to \infty)\]
as random variables taking values in the space of decreasing sequences of continuous functions
of $t \in \mathbb{R}$; in other words,
\[ {_M}\lambda_j(t) \overset{\mathrm{distr}}{\to} \lambda_j(0, t) \quad (M \to \infty)\]
as random continuous functions (for any $j \geq 1$).
\end{thm}
The unimodularity condition (\ref{eq:unimod}) is assumed to simplify the proof (see (\ref{eq:cheb}) below). In \cite[Part~III]{FS}, a (somewhat technical) argument allowing
to drop the assumption (\ref{eq:unimod}) was 
presented. Another argument allowing to
work with distributions violating (\ref{eq:unimod})
was developed by Erd\H{o}s and Knowles \cite{EK}
in the more general setting of random band matrices. Either of
these approaches could
probably be adapted to the current setting as well.

\vspace{2mm}\noindent
Finally, we mention two other possible extensions of
Theorems~\ref{thm:sosh}--\ref{thm:conv}. First,
in addition to $\beta = 1,2$, one may consider the universality class of symplectic
symmetry, $\beta = 4$. Second, instead of the corners
$H^{(N)}$, one may consider general sequences of
nested submatrices as in the work of Borodin 
\cite{B1,B2}. We expect  the methods of the
current paper to be adequate for these extensions. 

\subsection{Previous results} Let us place our results in context, and state some remarks and questions.

\paragraph{Fixed $s, t$, Gaussian invariant ensembles} For fixed $s, t$ (e.g.\ $s = t = 0$), the point processes
\[ _M\!\!\invb{\Lambda}(0, 0) = ( {_M\!\!\invb{\lambda}\!\!_1} (0, 0) \geq {_M\!\!\invb{\lambda}\!\!_2} (0, 0) \geq \cdots )\quad
(\beta=1,2) \]
associated to GOE (\ref{eq:GOE}) and GUE (\ref{eq:GUE}) were studied by Tracy and Widom \cite{TW1,TW2} and by Forrester \cite{F} (along with $\beta=4$ not discussed in this paper); they showed
that the sequence $(_M\!\!\invb{\Lambda}(0, 0))_M$ has a distributional limit,  
called the Airy$_\beta$ point
process. For $\beta = 2$, the Airy$_\beta$ point process is a determinantal
point process on $\mathbb{R}$ (see Johansson \cite{Joh} for definitions) with kernel
\[ A(\lambda_1, \lambda_2) = \int_0^\infty \Ai(\lambda_1 + u) \Ai(\lambda_2 + u) du~. \]
For $\beta = 1$, the Airy$_\beta$ point process is a Pfaffian point process, see \cite{TW2}
and the book of Mehta \cite{Mehta}.

The distribution of the largest eigenvalue of the Airy$_\beta$ point process (for $\beta=1,2,4$) is named the Tracy--Widom distribution
in honour of Tracy and Widom, who expressed it in terms of the Hastings--McLeod
solution to the Painlev\'e-II equation.

\paragraph{Fixed $s,t$, general Wigner matrices} In \cite{Sosh}, Soshnikov proved a general limit
theorem (universality theorem), stating that  the results of \cite{TW1,TW2,F} quoted above
are valid for any random Hermitian matrix $H$ satisfying a1)--a3). Namely, if $H$ satisfies a1), a2), and a3$_1$), then the point
process obtained by edge scaling converges to the Airy$_1$ point process, whereas if $H$ 
satisfies a1), a2), and a3$_2$), the limit is the Airy$_2$ point process. We refer
for example to the book of Pastur and Shcherbina \cite[Section~1.3]{PShch} for
a discussion of universality in random matrix theory.

The proof of Soshnikov is based on the moment method, see further Section~\ref{s:sosh}. An additional proof,
using a modification of the moment method involving Chebyshev polynomials,  was given in \cite{FS}; see further Section~\ref{s:prop}.

The assumptions a1) and a2) were gradually  relaxed by Ruzmaikina \cite{R}, Khorunzhiy \cite{Kh},
and finally Lee and Yin \cite{LY}, who showed that, in the case when the off-diagonal entries are identically 
distributed, have zero mean, and satisfy a3$_\beta$), the necessary and sufficient condition for convergence to the 
Airy$_\beta$ point process 
is given by
\[ \mathbb{P} \left\{ |H(1, 2)| \geq R \right\} = o(R^{-4})~, \quad R \to +\infty~.\]

It may be interesting to extend Theorem~\ref{thm:sosh} to the generality of the results of Lee and Yin \cite{LY}.

\vspace{3mm}\noindent For varying $s,t$, we are mainly aware of previous work pertaining to the invariant ensembles
$\ltau\invb{H}$ of Example~1.

\paragraph{Time-dependent matrix (fixed $t$, varying $s$) with unitary symmetry ($\beta=2$)} The process $_M\!\!\invu\Lambda(s,0)$ was
studied by Mac\^edo \cite{Mac} and by Forrester, Nagao,
and Honner \cite{FNH}; see also the book of Forrester \cite[7.1.5]{Fbook}. They introduced
the extended Airy kernel\footnote{in a slightly different form; we follow the
conventions of Johansson \cite{Joh}}
\begin{equation}\label{eq:airyker}\begin{split}
&A(s_1, \lambda_1; s_1, \lambda_2)\\
&\qquad=
\begin{cases}
\int_0^\infty e^{-u(s_1 - s_2)}\Ai(\lambda_1 + u)\Ai(\lambda_2 + u)
du, &s_1 \geq s_2\\
-\int_{-\infty}^0 e^{-u(s_1 - s_2)}\Ai(\lambda_1+u)\Ai(\lambda_2+u)
du,  &s_1 < s_2,
\end{cases}\end{split}\end{equation}
and showed that the finite-dimensional marginals of $_M\!\!\invu\Lambda(s,0)$ converge (as $M \to \infty$)
to those of the Airy$_2$ time-dependent point process, a determinantal process with kernel given by (\ref{eq:airyker}).

Comparing the result of \cite{Mac,FNH} with Theorem~\ref{thm:sosh}, we see that $\AD_2(s, 0)$
has the distribution of the Airy$_2$ time-dependent
point process. Theorem~\ref{thm:sosh} shows
that the results of \cite{Mac,FNH} remain valid in
the general setting of time-dependent Wigner
matrices.

\vspace{2mm}\noindent
The Airy$_2$ time-dependent point process was further studied by Pr\"ahofer and Spohn \cite{PS}, who
identified it as a limit process for a certain random growth model (polynuclear growth). It
is conjectured to be a universal limit for a wide class of random growth models in
the Kardar--Parisi--Zhang universality class. For several models, the convergence
to the Airy$_2$ process was rigorously established; see especially Johansson \cite{JohPNG,Joh}.
Subsequent developments, too numerous to be discussed here, are surveyed by Borodin--Gorin
in \cite{BGsurv} and by Corwin in \cite{Cor}.

Pr\"ahofer and Spohn \cite{PS} also proved that the Airy$_2$ time-dependent point process
has a modification with continuous paths; an additional proof was given by
Johansson \cite{JohPNG} (note that, although the results of \cite{PS,JohPNG} are stated for
the right-most point of the process, the proofs apply to a point with any fixed index). A detailed
study of the Airy$_2$ time-dependent point process was performed by Corwin and 
Hammond \cite{CH}, who coined the term ``Airy line ensemble'' and established, among
other results, local absolute continuity with respect to the Brownian motion, which also implies
the existence of a continuous modification. The second part of Theorem~\ref{thm:prop} can be seen
as an extension of the continuous modification to the two-parameter process $\AD_2(s,t)$, and also
as a generalisation to the case $\beta=1$.

\paragraph{Corners of a fixed matrix (fixed $s$, varying $t$), unitary symmetry ($\beta=2$)}
For fixed $\tau$ (say, $\tau = 0$), the process $^{(0)}\invu\Xi$  is known
as the GUE corner process. Its distribution can be explicitly written down,
see Gelfand and Naimark \cite{GN}, Baryshnikov \cite{Bar},
and Neretin \cite{N}. Johansson and Nordenstam \cite{JN} and Okounkov and
Reshetikhin \cite{OR} showed that the GUE corner process is determinantal, and
derived explicit expressions for the kernel. We refer to 
\cite{JN,OR}, and the  work of Gorin \cite{G}
(and references therein) for a discussion of the GUE
corner process as a universal  scaling limit for
random tilings where the disordered region touches
the boundary.

The edge scaling $_M\invu\Lambda(0, t)$ was studied by Forrester and Nagao \cite{FN}
(see also the book of Forrester \cite[Section~11.7.2]{Fbook}), who proved the convergence
to the Airy$_2$ time-dependent point process in the sense of finite-dimensional
marginals.

Theorem~\ref{thm:sosh} shows
that the results of \cite{FN} remain valid in
the general setting of Wigner
corners.

\paragraph{Time-dependent minors (varying $s,t$), unitary symmetry ($\beta = 2$)} The time-dependent process $^{(\tau)}\invu\Xi$ was studied by Adler, Nordenstam,
and van Moerbeke \cite{ANvM} and by Ferrari and Frings \cite{FF}, who showed that the process is determinantal along space-like
paths (in the terminology of Borodin and Ferrari \cite{BF}, i.e.\
$\tau_1 \leq \cdots \leq \tau_k$ and $N_1 \geq \cdots \geq N_k$). In \cite[Section~4.4]{Fer}, Ferrari
showed that the edge scaling limit along space-like paths is the Airy$_2$ time-dependent
point process (note that a space-like path lies on an $\ell_1$ geodesic,
thus this result is consistent with the first part of Theorem~\ref{thm:prop}).

It would be interesting to derive an explicit connection between the extended Airy kernel
(\ref{eq:airyker}) and the formul{\ae} (\ref{eq:defpsi1sh}), (\ref{eq:defpsi2sh}), (\ref{eq:defpsibeta}),
(\ref{eq:sh}) of Section~\ref{s:prop}.

\paragraph{Orthogonal symmetry ($\beta =1$)} Much less is known for $\beta = 1$, even for $\ltau\invo{H}$. It was conjectured
that the distribution of the right-most point of the process $_M\!\!\invo\Lambda(s, 0)$ converges to the Airy$_1$  process introduced by Sasamoto \cite{Sas}. However, the
numerical evidence of Bornemann, Ferrari, and Pr\"ahofer \cite{BFP} strongly suggests
that this conjecture is false. See further Ferrari \cite{Fer} for a discussion.

To the best of our knowledge, the only explicit description of the scaling limit of  
$_M\!\!\invo\Lambda(s, 0)$ which is currently available is the one of Section~\ref{s:prop}.

\paragraph{Global regime: fluctuations of linear statistics} 
In the  works \cite{B1,B2}, Borodin studied the eigenvalue
statistics of time-dependent Hermitian random matrices (with $\beta$ being either
$1$ or $2$), who showed that the global fluctuations of $\ltau\Xi$ (and of eigenvalues
of more general submatrices of $\ltau{H}$) are given by a certain family of
correlated  Gaussian free fields. The results of \cite{B1,B2} do not require symmetric distribution
A1) nor subgaussian tails A2); the distribution of the entries is assumed to have finite moments,
and vanishing fourth cumulant. Kargin \cite{Kar} extended the results of  \cite{B1} to the case
when the fourth cumulant does not vanish.

\paragraph{Extension to $\beta \neq 1,2$} Borodin and Gorin \cite{BG} constructed
an extension of the GUE and GOE corner processes (in fact, of general Jacobi corner processes) to
 general $\beta > 0$, and
proved a limit theorem for fluctuations of
eigenvalue statistics in this general setting. 

We also mention two other possible approaches to the extension of time-dependent Wigner corner
processes to arbitrary $\beta$: the ghosts and shadows formalism of Edelman \cite{E}, which
is yet to find a rigorous justification, and the matrix model of $\beta$-Dyson Brownian motion put
forth by Allez, Bouchaud,
and Guionnet \cite{ABG,AG}. 

\subsection{Plan of the following sections}
The proofs of the results are based on the moment method. Following Soshnikov \cite{Sosh},
we consider the mixed moments
\begin{equation}\label{eq:mixedmoments}
\mathcal{M}(m_1, \cdots, m_k; \tau_1, \cdots, \tau_k; N_1, \cdots, N_k)
= \mathbb{E} \prod_{p=1}^k \tr \left( \frac{^{(\tau_p)}{H}^{(N_p)}}{2 \sqrt{N_p}}\right)^{m_p}~.
\end{equation}
In the asymptotic regime
\[ m_p \sim 2 \alpha_p M^{2/3}~,\quad \tau_p \sim s_p M^{-1/3}~,
\quad N_p - M \sim 2 \tau_p M^{-1/3}~, \]
(\ref{eq:mixedmoments}) is an approximation to the Laplace transform
\[ \mathbb{E} \prod_{p=1}^k \sum_j \left\{  \exp(\alpha_p \, _M\!\lambda_j(s_p, t_p)) + 
(-1)^{m_p} \exp(\alpha_p \, _M\!\lambda_j^\sim(s_p, t_p))\right\}~, \] 
where $_M\!\lambda_j^\sim$ are constructed (\ref{eq:scaling}) from the eigenvalues of
the matrix $-H$. In Section~\ref{s:sosh}, we
show that the asymptotics of (\ref{eq:mixedmoments}) does not depend on the choice of a matrix
$H$ satisfying A1)--A3). This statement is derived from the estimates of \cite{Sosh}.

In Section~\ref{s:prop}, we focus on the special case of  matrices with unimodular entries (\ref{eq:unimod}).
In this case, the combinatorial approach of \cite{FS}, based on a modification of the moment method
involving Chebyshev polynomials, takes a particularly simple form; we apply
it to find the asymptotics of (\ref{eq:mixedmoments}). In principle, the arguments of \cite{FS}
could be applied directly to the general case A1)--A3); however, the reduction to (\ref{eq:unimod})
using the estimates available from \cite{Sosh} seems to be simpler.

An additional set of combinatorial estimates required for the proof of the second half of Theorem~\ref{thm:prop} and of Theorem~\ref{thm:conv} is established in Section~\ref{s:conv}.

In Section~\ref{s:pf}, Theorems~\ref{thm:sosh}, \ref{thm:prop}, and \ref{thm:conv} are 
derived from the results of Sections~\ref{s:sosh}, \ref{s:prop}, and \ref{s:conv}. 

Section~\ref{s:pf} is preceded by Section~\ref{s:top}, collecting preliminaries pertaining to 
parameter-dependent point processes. The content of Section~\ref{s:pf} is  
well-known to experts (in particular, some of the arguments do not differ much from those developed
in \cite{Sosh} for a single point process), and is recounted  for convenience of reference.

\paragraph{Terminology}
In the table below, we summarise some of the common terminology and compare it with our notation. The letters $\AD$ stand for G.~B.~Airy and F.~J.~Dyson.
\FloatBarrier
\begin{table}
 \begin{tabular}{lll}
 {\bf Our notation} & \parbox[t]{2in}{\bf the process \\ \{and its modification\}} & \parbox[t]{1in}{\bf the right\-most point} \\ \\
 \hline \hline
 $\AD_2(0,0)$ & Airy$_2$ point process & Tracy-Widom$_2$ \\ \hline
 $\AD_2(s,0) \eqD \AD_2(0, s)$ & \parbox[t]{2in}{time-dependent Airy$_2$ \\ point process = multiline \\ Airy$_2$ process \\ \{Airy$_2$ line ensemble\}} & Airy$_2$ process \\ \hline
 $\AD_1(0, 0)$ & Airy$_1$ point process & Tracy-Widom$_1$ \\ \hline
 $\AD_1(s,0) \eqD \AD_1(0, s)$ & \parbox[t]{2in}{?} & \parbox[t]{1.2in}{? (not the Airy$_1$\\ process of \cite{Sas})}\\ \hline
 \end{tabular}
\end{table}
\FloatBarrier

\paragraph{Notation} Throughout this text, $C$
denotes positive constants the value of which may
vary from line to line; $C_{[\ast]}$ will denote a positive number
depending only on the quantity $[\ast]$. The symbol $\vee$ stands for maximum,
and $\wedge$ stands for minimum.

We hope the reader will forgive
the multitude of indices, which we try to use
consistently:  the
time ($\tau$ or its scaling $s$) is placed in
the left superscript, the number of the submatrix ($N$ or
its scaling $t$) -- in the right one; the large parameter
is placed in the left subscript, and, finally, the
right subscript is reserved for the index of an element
in a decreasing sequence.

\paragraph{Acknowledgment}  This work was inspired by a series of talks given by
Alexei Borodin at the IAS, and owes a lot to our
subsequent discussions and correspondence; I thank him very much. I am  
grateful to  Ivan Corwin, Ohad Feldheim, Patrik Ferrari, Vadim Gorin, Kurt Johansson,
and Jeremy Quastel for helpful suggestions and comments.

\section{Reduction to unimodular entries}\label{s:sosh}

Consider the mixed moments
\[\begin{split} \mathcal{M}(\bar{m}, \bar{\tau}, \bar{N})
&= \mathcal{M}(m_1, \cdots, m_k; \tau_1, \cdots, \tau_k; N_1, \cdots, N_k)\\
&=\mathbb{E} \prod_{p=1}^k \tr \left(\frac{^{(\tau_p)}{H}^{(N_p)}}{2 \sqrt{N_p}}\right)^{m_p}~.
  \end{split} \]
One can express $\mathcal{M}(\bar{m}, \bar{\tau}, \bar{N})$ as a sum
\begin{equation}\label{eq:paths} \mathcal{M}(\bar{m}, \bar{\tau}, \bar{N})
= \sum \mathbb{E} \prod_{p=1}^{k} \left[ N_p^{-m_p/2} {\prod_{i=1}^{m_p}} {^{(\tau_p)}}H^{(N_p)}(u_{p, i}, u_{p, i+1}) \right]
\end{equation}
over $k$-tuples
\begin{equation}\label{eq:ktuple}
u_{1,0}u_{1,1}u_{1,2}\cdots u_{1,m_1}; u_{2,0}u_{2,1}u_{2,2}\cdots u_{2,m_2};  \cdots; u_{k,0}u_{k,1}u_{k,2}\cdots u_{k,m_k} 
\end{equation}
such that $u_{p,i} \in \{1, \cdots, N_p\}$ ($1 \leq p \leq k$, $0 \leq i \leq m_p$) and 
$u_{p, m_p} = u_{p, 0}$ ($1 \leq p \leq k$).  Such a $k$-tuple is interpreted as
a $k$-tuple of closed paths in the complete graph with loops:
\[ K_\infty^+ = (V_\infty, E_\infty^+)~, \quad
V_\infty = \mathbb{N}~, \quad
E_\infty^+ = \left\{ (i, j) \, \mid \, (i, j) \in \mathbb{N} \right\}\]

Due to A1), only $k$-tuples in which every  (non-oriented) edge is traversed an even number of times give a non-zero contribution to (\ref{eq:paths}).

In this section, we prove
\begin{lemma}\label{c:sosh}
For $M/2 \leq N_1, \cdots, N_k \leq 2M$, $-1 \leq \tau_1, \cdots, \tau_k \leq 1$, and arbitrary
$m_1, \cdots, m_k$, the following estimates hold:
\begin{enumerate}
\item $\mathcal{M}(\bar{m}, \bar{\tau}, \bar{N}) \leq \prod_{p=1}^k \frac{CM}{m_p^{3/2}} \exp(C_km_p^3/M^2)$;
\item for $m_p = O(M^{2/3})$, the sub-sum of (\ref{eq:paths}) corresponding to $k$-tuples of paths
in which at least one non-oriented edge is traversed more than twice or at least one path has a loop
(i.e.\ $u_{p, i} = u_{p, i+1}$ for some $p, i$) vanishes in the limit $M \to \infty$.
\end{enumerate}
\end{lemma}

We derive Lemma~\ref{c:sosh} from its special case proved in \cite[Theorem~3]{Sosh} (see 
further \cite[Section~I.5]{FS}):
\begin{lemma}\label{l:sosh}
For $N_1 = \cdots = N_k = M$ , $\tau_1 = \cdots = \tau_k$, and arbitrary $m_1, \cdots, m_k$, the
following estimates hold:
\begin{enumerate}
\item $\mathcal{M}(\bar{m}, \bar{\tau}, \bar{N}) \leq \prod_{p=1}^k \frac{CM}{m_p^{3/2}} \exp(C_km_p^3/M^2)$,
where the numbers $C>0$ and $C_k > 0$ may depend on $C_0$ from A2);
\item for $m_p = O(M^{2/3})$, the sub-sum of (\ref{eq:paths}) corresponding to $k$-tuples of paths
in which at least one non-oriented edge is traversed more than twice or at least one path has a loop
(i.e.\ $u_{p, i} = u_{p, i+1}$ for some $p, i$) vanishes in the limit $M \to \infty$.
\end{enumerate}
\end{lemma}

\begin{proof}[Proof of Lemma~\ref{c:sosh}]
By H\"older's inequality and the first item of Lemma~\ref{l:sosh},
\[\begin{split}
 \mathcal{M}(\bar{m}, \bar{\tau}, \bar{N})
 &\leq \prod_{p=1}^k \left[ \mathbb{E} \left\{ \tr \left(\frac{^{(\tau_p)}{H}^{(N_p)}}{2 \sqrt{N_p}}\right)^{m_p} \right\}^{2k} \right]^{\frac{1}{2k}} \\
 &\leq  \prod_{p=1}^k \frac{CM}{m_p^{3/2}} \exp \left( C_k m_p^3 \, \big/ \, M^2 \right)~.
 \end{split}\]
 This proves item~1 of Lemma~\ref{c:sosh}. To prove the second
item, we need the following argument, which we
also use in the next section. Without loss of generality we may
assume that $N_1 \leq \cdots \leq N_p$. Divide the $k$-tuples of paths (\ref{eq:ktuple}) into
isomorphism classes as follows: $(u_{p,i}) \sim (u'_{p,i})$
if there exists a permutation $\pi \in S_\infty$ of $\mathbb{N}$ such that 
$u_{p,i} = \pi(u'_{p,i})$ ($1 \leq p \leq k$, $1 \leq i \leq m_p$). The value of
\[ \mathbb{E} \prod_{p=1}^k \prod_{i=1}^{m_p}  {^{(\tau_p)}}H^{(N_p)}(u_{p, i}, u_{p, i+1}) \]
is constant on every isomorphism class. It does not depend on $N_p$ ($1 \leq p \leq k$);
as a function of $\tau_p$ ($1 \leq p \leq k$), it is maximal when all $\tau_p$ coincide. Next,
the quantity 
\[ \prod_{p=1}^k N_p^{-m_p/2} \, \# (\text{isomorphism class})\]
is bounded from above by and asymptotic to
\begin{equation}\label{eq:contriso}
\prod_{p=1}^k N_p^{-m_p/2} \prod_{\text{vertices $v$}} 
\min \left\{ N_p \, \mid \, {1 \leq p \leq k~, \,\, \text{v lies on the $p$-th path} } \right\}~,
\end{equation}
since $m_p = O(M^{2/3}) = o(M^{1/2})$. The expression (\ref{eq:contriso})
is in turn bounded by $\mathrm{const}^k$ times its value when all $N_p$ are set to
$M/2$. Therefore we conclude that the contribution of an isomorphism
class to $\mathcal{M}(\bar{n}, \bar{\tau},\bar{N})$
is bounded by a constant (dependent only on $k$)
times its contribution to
\[ \mathcal{M}(m_1, \cdots, m_k; 0, \cdots, 0; M/2, \cdots, M/2)~.\]
Hence the second item of Lemma~\ref{c:sosh} follows from its
counterpart in Lemma~\ref{l:sosh}.
\end{proof}

\section{Asymptotics of mixed moments}\label{s:prop}

Next we turn to
\begin{lemma}\label{l:sprop} Let $\beta \in \{1,2\}$. There exists continuous
functions 
\[ \phi_\beta^\#: \mathbb{R}_+^k \times \mathbb{R}^k \times \mathbb{R}^k \to \mathbb{R}_+\]
so that for
\[\begin{split} &m_p \sim \alpha_p M^{2/3}, \quad \tau_p = s_p M^{-1/3}~, \quad N_p = M(1 + 2t_p M^{-1/3}) \quad (1 \leq p \leq k)~, \\
&\sum_{p=1}^k m_p \equiv 0 \mod 2 \end{split} \]
one has:
\begin{equation}\label{eq:sprop} \mathcal{M}(\bar{m}, \bar{\tau}, \bar{N}) = \phi_\beta^\#(\bar{\alpha}, \bar{s}, \bar{\tau}) + o(1)~,\end{equation}
and the function $\phi_\beta^\#(\bar{\alpha}, \bar{s}, \bar{t})$ depends on $\bar{s}$ and $\bar{t}$
only via the $\ell_1$ distances $|s_p - s_r|+|t_p-t_r|$ ($1 \leq p < r \leq k$).
\end{lemma}

According to (both items of) Lemma~\ref{c:sosh}, it is sufficient to prove the lemma in our favourite special case,
and we pick the matrices with unimodular entries (\ref{eq:unimod}). In this case, the combinatorial
arguments of \cite{FS} are especially transparent, due to the following identity
(\ref{eq:cheb})  (see e.g.\ \cite[Claim~II.1.2]{FS}).

Denote
\[ P_n^{(N)}(\lambda) = U_n\left( \frac{\lambda}{2\sqrt{N-2}} \right) - \frac{1}{N-2} U_{n-2} \left( \frac{\lambda}{2\sqrt{N-2}} \right)~, \]
where
\[ U_n(\cos \theta) = \frac{\sin((n+1)\theta)}{\sin \theta} \]
are the Chebyshev polynomials of the second kind, and we formally set $U_{-2} \equiv U_{-1} \equiv 0$.
Then, for any Hermitian matrix $H$ satisfying
the unimodularity assumption
\begin{equation}\label{eq:unimod'}
|{H}_{i,j}| = 1 \quad (i \neq j)~, \quad H(i, i) = 0 \quad (1 \leq i, j)~,
\end{equation}
we have (see e.g.\ \cite[Claim~II.1.2]{FS}):
\begin{equation}\label{eq:cheb}
\tr P_n^{(N)} \left( H^{(N)} \right) = \frac{1}{(N-2)^{n/2}} \sum \prod_{i=1}^n H(u_i, u_{i+1})~,
\end{equation}
where the sum is over paths $u_0u_1\cdots u_n$ in $\{1,\cdots,N\}$ such that $u_n = u_0$ (the path is closed), $u_j \neq u_{j+1}$ (no loops), and
$u_j \neq u_{j+2}$ (no backtracking). Define the modified mixed moments $\widetilde{\mathcal{M}}(\bar{n}, \bar{\tau}, \bar{N})$
corresponding to a time-dependent Hermitian random matrix $\ltau{H}$ as follows:
\[
\widetilde{\mathcal{M}}(\bar{n}, \bar{N}, \bar{\tau}) = \mathbb{E} \prod_{p=1}^k \tr P_{n_p}^{(N_p)} \left( ^{(\tau_p)}H^{(N_p)} \right)~.\]
If $\ltau{H}$ satisfies the unimodularity
assumption (\ref{eq:unimod}), we have from (\ref{eq:cheb}):
\begin{equation}\label{eq:cpaths}
\widetilde{\mathcal{M}}(\bar{n}, \bar{N}, \bar{\tau})
= \sum  \mathbb{E} \prod_{p=1}^{k} \left[ (N_p-2)^{-n_p/2} {\prod_{i=1}^{n_p}} {^{(\tau_p)}}H(u_{p, i}, u_{p, i+1}) \right]~,
\end{equation}
where the sum is over $k$-tuples of closed non-backtracking paths without loops
\[ u_{1,0}u_{1,1}u_{1,2}\cdots u_{1,n_1}; u_{2,0}u_{2,1}u_{2,2}\cdots u_{2,n_2};  \cdots; u_{k,0}u_{k,1}u_{k,2}\cdots u_{k,n_k}\]
such that $u_{p,i} \in \{1,\cdots,N_p\}$. As in the previous section, we may only consider even $k$-tuples, i.e.\ $k$-tuples in which every edge is traversed
an even number of times. Also, we only consider the tuples $\bar{n}$ for which the sum $\sum_p n_p$ is even, otherwise our set
of $k$-tuples is empty and $\widetilde{\mathcal{M}}=0$.

We shall be interested in the asymptotics of $\widetilde{\mathcal{M}}(\bar{n}, \bar{N}, \bar{\tau})$ for
\begin{equation}\label{eq:regime}\begin{split} 
&n_p \sim \alpha_p M^{1/3}~, \quad \tau_p = s_p M^{-1/3}~, \quad  N_p = M(1 + 2t_p M^{-1/3})~, \\
&\sum_{p=1}^k n_p \equiv 0 \mod 2~.
\end{split}
\end{equation}

\begin{lemma}\label{l:sprop1} 
For  matrices with unimodular entries (\ref{eq:unimod'}), we have in the asymptotic regime (\ref{eq:regime}):
\[ \widetilde{\mathcal{M}}(\bar{n},\bar{N},\bar{\tau}) = M^{k/3}
(\psi_\beta^\#(\bar{\alpha}, \bar{s}, \bar{t}) + o(1))~,\]
where $\psi_\beta^\#(\bar{\alpha}, \bar{s}, \bar{t})$
depend on $\bar{s}$ and $\bar{t}$ only via the $\ell_1$ distances
\[ \{ |s_p - s_r|+ |t_p - t_r| \}_{1 \leq p < r \leq k}~.\]
\end{lemma}

\begin{proof}[Proof of Lemma~\ref{l:sprop} using Lemma~\ref{l:sprop1}]

We follow the strategy of \cite[Section~I.5]{FS}. In view of Lemma~\ref{c:sosh} (item 2), we can
assume that the matrix $H$ has unimodular entries (\ref{eq:unimod'}); in this case Lemma~\ref{l:sprop1}
is applicable. 

Recall the identities (cf.\ Snyder \cite{Sn})
\begin{equation}\label{eq:snyder} \begin{split}
\lambda^{2m} &= \frac{1}{(2m+1)2^{2m}} \sum_{n=0}^m (2n+1) 
\binom{2m+1}{m-n} U_{2n}(\lambda)~; \\
\lambda^{2m+1} &= \frac{1}{2m 2^{2m-1}} \sum_{n=0}^m 2n \binom{2m}{m-n}
U_{2n-1}(\lambda)
\end{split}\end{equation}
expressing $\lambda^{2m}$ and $\lambda^{2m+1}$ in terms of the Chebyshev
polynomials of the second kind, and note that
\begin{equation}\label{eq:geom}
\begin{split}
U_n &= \left\{ U_n - \frac{1}{N-2} U_{n-2} \right\} 
+ \frac{1}{N-2} \left\{ U_{n-2} - \frac{1}{N-2} U_{n-4} \right\} \\
&+ \frac{1}{(N-2)^2} \left\{ U_{n-4} - \frac{1}{N-2} U_{n-6} \right\} + 
\cdots~.
\end{split}
\end{equation}

Next, we make use of the following estimate from \cite[Theorem~I.5.3, Proposition~II.3.4]{FS}:
\begin{lemma}\label{l:fs} Fix $k \geq 1$, and let $\bar{n}$ be
a $k$-tuple of indices. Then
\[ \widetilde{\mathcal{M}}(\bar{n},\bar{N},\bar{\tau})
\leq (Cn)^k \exp \left\{ C_k n^{3/2} M^{-1/2}\right\}~.\]
If $n_p = O(M^{1/3})$ ($1 \leq p \leq k$),
 the contribution of $k$-tuples of paths in which at least
one edge is traversed more than twice is $o(M^{k/3})$.
\end{lemma}
The estimate was proved in \cite{FS} for the case
of identical $\tau$'s and $N$'s, and extends to the
general case by the argument of Section~\ref{s:sosh}.

Now we express the mixed moments $\mathcal{M}$
in terms of the modified moments $\widetilde{\mathcal{M}}$
using (\ref{eq:snyder}) and (\ref{eq:geom}), and pass to the
limit as $M \to \infty$. We obtain (\ref{eq:sprop})
with 
\begin{equation}\label{eq:phipsi} \phi_\beta^\#(\bar{\alpha}, \bar{s}, \bar{t}) = 
\prod_{p=1}^k \int_0^\infty \frac{2 \xi_p d\xi_p}{\sqrt{\pi \alpha_p}} e^{-\xi_p^2}
\psi_\beta^\#(2\sqrt{\bar{\alpha}}\bar{\xi}, \bar{s}, \bar{t})~,\end{equation}
where $\sqrt{\bar{\alpha}}\bar{\xi}$ denotes coordinate-wise product.
To justify the passage to the limit, we first split  the  sums  (\ref{eq:snyder})
in two parts: the convergence of the dominant part  $\delta M^{1/3} \leq n \leq
\delta^{-1} M^{1/3}$ to the integral as $M \to \infty$ and then $\delta \to +0$ 
follows from Lemma~\ref{l:sprop1} and the first part of Lemma~\ref{l:fs},
whereas the convergence of the subdominant part to zero follows from the 
first part of Lemma~\ref{l:fs} (cf.\ \cite[Proof of Theorem~I.2.4]{FS}).

Thus we conclude the proof of Lemma~\ref{l:sprop}.
\end{proof}

\begin{proof}[Proof of Lemma~\ref{l:sprop1}]

Let us first consider the case $\beta = 1$. In \cite{FS},  even $k$-tuples of closed non-backtracking 
paths without loops were divided (up to
an asymptotically negligible fraction) into topological equivalence classes, 
called $k$-diagrams. We refer to \cite[Definition~2.3, Remark~2.4]{band} for precise definitions; here we content ourselves with
Figures~\ref{fig:diag1} and \ref{fig:diag2} (which we copy from \cite{band}, correcting a mistake in the caption),
and an example: the pair of paths
\begin{equation}\label{eq:example}
1 \, 2 \, 3 \, 4 \, 5 \, 3 \, 4 \, 5\, 3 \, 2 \, 6 \, 7 \, 8 \, 9 \, 7 \, 6 \, 2 \, 1~, \quad
 10 \, 11 \, 8 \, 7 \, 9 \, 8 \, 11 \, 10 
\end{equation}
corresponds to the rightmost diagram of Figure~\ref{fig:diag2}.

\begin{figure}[h]
\vspace{2cm}
\setlength{\unitlength}{1cm}
\begin{pspicture}(-1,0)
\psline(.2,1)(2, 1)
\psarcn(2.8,1){.8}{180}{195}
\psset{arrows=-}
\psline(2,.8)(2.1,.92)
\psarcn(2.8,1){.7}{185}{195}
\psline(2,.92)(2.1,.8)
\psset{arrows=->}
\psline(2,.9)(.2,.9) %
\psline(4.2,1)(6,1)
\psarcn(6.8,1){.8}{180}{185}
\psline(6,.9)(5,.9)
\pscurve(5,.9)(5.5,.7)(6,.5)(6.4,.5)
\psarc(6.8,1){.65}{230}{218}
\pscurve(6.27,.6)(5.9, .6)(5,.8)(4.8,.9)(4.2,.9) %
\psline(8.2,1)(9, 1)
\psline(9, 1)(9.2, 1.7)
\psarcn(9.3, 2.05){.35}{255}{270}
\psarcn(9.3, 2.05){.45}{265}{270}
\psline(9.3, 1.7)(9.1, 1)
\psline(9.1, 1)(10, 1)
\psarcn(10.8,1){.8}{180}{195}
\psset{arrows=-}
\psline(10,.8)(10.1,.92)
\psarcn(10.8,1){.7}{185}{195}
\psline(10,.92)(10.1,.8)
\psset{arrows=->}
\psline(10,.9)(8.2,.9)
\end{pspicture}
\caption[Fig. 4.5]{A few simple 1-diagrams: $s = 1$ (left),
    $s = 2$ (centre, right)}\label{fig:diag1}
\vspace{2mm}
\end{figure}
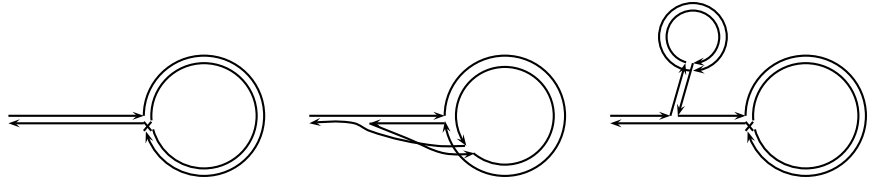

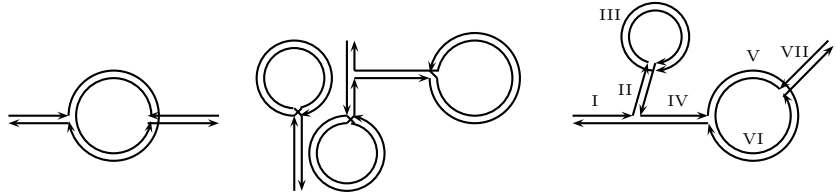
\begin{figure}[h]
\vspace{3cm}
\setlength{\unitlength}{1cm}
\begin{pspicture}(-1,0)
\psline(.2,1)(1, 1)
\psarcn(1.6,1){.6}{180}{187}
\psline(1,.9)(.2,.9)
\psline(3,1)(2.05, 1)
\psarc(1.6,1){.5}{0}{352}
\psline(2.05, .9)(3,.9)
\psline(4,0)(4, 1)
\psarcn(4,1.5){.5}{270}{278}
\psset{arrows=-}
\psline(4.1,1)(4,1.1)
\psarcn(4,1.5){.4}{268}{282}
\psline(4.1,1.1)(4,1)
\psset{arrows=->}
\psline(4.1,1)(4.1,0)
\psline(4.7,2)(4.7, 1)
\psarc(4.7,.5){.5}{90}{82}
\psset{arrows=-}
\psline(4.8,1)(4.7,.9)
\psarc(4.7,.5){.4}{90}{82}
\psline(4.8,.9)(4.7,1)
\psset{arrows=->}
\psline(4.8,1)(4.8,1.5)
\psline(4.8,1.5)(5.8, 1.5)
\psarc(6.4,1.5){.6}{180}{173}
\psset{arrows=-}
\psline(5.8, 1.6)(5.9, 1.5)
\psarc(6.4,1.5){.5}{180}{170}
\psline(5.9,1.6)(4.8, 1.6)
\psset{arrows=->}
\psline(4.8,1.6)(4.8,2)
\psline(7.7,1)(8.5, 1)
\rput(8, 1.2){\tiny I}
\psline(8.5, 1)(8.7, 1.7)
\rput(8.4, 1.35){\tiny II}
\psarcn(8.8, 2.05){.35}{255}{270}
\psarcn(8.8, 2.05){.45}{265}{270}
\rput(8.2, 2.35){\tiny III}
\psline(8.8, 1.7)(8.6, 1)
\psline(8.6, 1)(9.5, 1)
\rput(9.1, 1.2){\tiny IV}
\psarcn(10.1,1){.6}{180}{190}
\rput(10.1, 1.8){\tiny V}
\rput(10.1, 0.7){\tiny VI}
\psline(9.5,.9)(7.7,.9)
\psline(11.1, 2)(10.45, 1.35)
\psarc(10.1,1){.5}{45}{37}
\psline(10.5, 1.25)(11.17, 1.93)
\rput(10.65, 1.85){\tiny VII}
\end{pspicture}
\caption[Fig. 4.5]{A few simple 2-diagrams: $s = 2$ (left),
    $s = 3$ (centre, right). The leftmost diagram is the only one which survives in the asymptotic
    regime considered by Borodin \cite{B1,B2}, and it gives rise to the family of Gaussian
    free fields put forth in these works.}\label{fig:diag2}
\end{figure}

To every $k$-diagram $\mathcal{D}$, one associates an integer $s \geq k$, and
a  multigraph $(V, E)$ with $\# V = 2s$ vertices and $\# E = 3s-k$ edges. For $p = 1, \cdots, k$ and $e \in E$,
let $c_p(e) \in \{0, 1, 2\}$ be the number of times an edge corresponding to $e$ is traversed by the $p$-th path
in the tuple. Consider the system of equations $\mathfrak{E}_\mathcal{D}(\bar{n})$:
\[ \sum_{e \in E} c_p(e) \ell(e) = n_p \quad (1 \leq p \leq k)~.\]
The number $\ell(e)$ represents the number of edges in the $k$-tuple of paths corresponding to the edge $k$
of the $k$-diagram. 

\vspace{1mm}
\noindent In the example (\ref{eq:example}) above, 
\[ \begin{split}
&c_1(I) = c_1(II) = c_1(II)=c_1(IV) =c_2(VII) = 2~,\\
&c_1(V) = c_1(VI) = c_2(V) = c_2(VI) = 1~, \\
&c_1(VII) = c_2(I) = c_2(II) = c_2(III) = c_2(IV) = 0~,
\end{split}\]
where the Roman numerals number the edges of the
diagram as on Figure~\ref{fig:diag2}, and
\[ \ell(I) = 1, \ell(II) = 1, \ell(III)=3, \ell(IV)=2,
\ell(V)=1, \ell(VI)=2, \ell(VII) = 2~. \]

Also, for every $e \in E$, let $1 \leq p_-(e) \leq p_+(e) \leq k$ be the indices of the two
paths traversing $e$. In the example (\ref{eq:example}), 
\[\begin{split}
&p_\pm(I)=p_\pm(II)=p_\pm(III)=p_\pm(IV) = 1~,\\
&p_-(V)=p_-(VI)=1~, \, p_+(V)=p_+(VI)=2~, \, p_\pm(VII)=2~. 
\end{split}\]
Finally, let
$\Delta_\mathcal{D}(\bar{n})$ be the convex polytope of positive real solutions to  $\mathfrak{E}_\mathcal{D}(\bar{n})$
(this is a $(3s-2k)$-dimensional polytope in $\mathbb{R}^{3s-k}$).

To evaluate the contribution of $\mathcal{D}$ to (\ref{eq:cpaths}) in the asymptotic
regime (\ref{eq:regime}), we need to sum $A(\bar{\ell})B(\bar{\ell})C(\bar{\ell})$ over
$\ell \in\Delta_\mathcal{D}(\bar{n})$, where
\[ A(\bar\ell) = \prod_{p=1}^k (N_p - 2)^{-n_p/2} = (1+o(1)) \prod_{p=1}^k N_p^{-n_p/2} 
= (1 + o(1)) \prod_{e} \left[ N_{p_+(e)}^{-\ell(e)/2}N_{p_-(e)}^{-\ell(e)/2}\right] \]
is the prefactor from (\ref{eq:cheb}); 
\[\begin{split}
 B(\bar\ell) 
 &= (1+o(1)) \prod_{v \in V} \min \left\{ N_j \, \mid \, \text{$v$ lies on the $j$-th path} \right\} \\
 &= (1+o(1)) M^{k-s} \prod_{e} \left( N_{p_+(e)} \wedge N_{p_-(e)} \right)^{\ell(e)}
\end{split}\]
is a combinatorial factor counting the number of ways to chose the vertices (the asymptotics
is valid for fixed $s$, whereas for large $s$ the same expression provides an upper bound);
and
\[ C(\bar\ell) = \prod_e \left( 1 - |\tau_{p_+(e)} - \tau_{p_-(e)}| + o(1) \right)^{\ell(e)} \]
takes the correlations of the matrix elements into account. One can see that the contribution
to the sum comes from $\bar\ell$ such that $\ell(e) \geq M^{1/3-\delta}$. Therefore
\begin{multline*}
\sum_{\bar\ell \in \Delta_{\mathcal{D}}(\bar{n})} A(\bar\ell) B(\bar\ell) C(\bar\ell) 
\\= \frac{(1 +o(1))}{M^{s-k}} \sum_{\bar\ell \in \Delta_{\mathcal{D}}(\bar{n})} \prod_{e \in E} \frac{\left(N_{p_+(e)} \wedge N_{p_-(e)}\right)^{\ell(e)}}{\left(N_{p_+(e)}  N_{p_-(e)}\right)^{\ell(e)/2}} 
\left( 1 - |\tau_{p_+(e)} - \tau_{p_-(e)}| \right)^{\ell(e)} 
\end{multline*}
which can be further rewritten as:
\[\frac{(1 + o(1))}{M^{s-k}} \,
\sum_{\bar{\ell} \in \Delta_\mathcal{D}(\bar{n})} \prod_{e \in E} 
\left( 1 - \frac{|N_{p_+(e)} - N_{p_-(e)}|}{M} \right)^{\frac{\ell(e)}{2}} 
(1 - |\tau_{p_+(e)} - \tau_{p_-(e)}|)^{\ell(e)}~.\]
Expressing $\bar{N}$,$\bar\tau$ in terms of $s, t$ (cf.\ (\ref{eq:regime}))
and replacing the Riemann sum with an integral, we finally obtain the asymptotic
expression
\[
(1+o(1)) M^{k/3} \!\!\!\!\!\int\limits_{\Delta_\mathcal{D}(\bar{\alpha})} \!\!\!d\bar{w}
\exp\left\{ - \sum_e  \left[ |t_{p_+(e)} - t_{p_-(e)}| + |s_{p_+(e)} - s_{p_-(e)}| \right] w(e)\right\}~,
\]
where the integral is with respect to the $(3s-2k)$-dimensional measure. Denoting the integral
by $I^\mathcal{D}\left(\bar{\alpha}, \bar{s}, \bar{t}\right)$
and setting 
\begin{equation}\label{eq:defpsi1sh} \psi_1^\#(\bar\alpha, \bar{s}, \bar{t}) = \sum_{\mathcal{D}} I^{\mathcal{D}}(\bar{\alpha},\bar{s},\bar{t})~,\end{equation} 
we obtain:
\begin{equation}\label{eq:formula} \widetilde{\mathcal{M}}(\bar{n}, \bar{N}, \bar{\tau})
= (1+o(1)) M^{k/3} \psi_1^\#(\bar{\alpha},\bar{s},\bar{t})~,\end{equation}
as claimed. 

To define $\psi_2^\#$, we say that a diagram $\mathcal{D}$ is orientable if every edge
is traversed in two opposite directions, For example, the diagrams of Figure~\ref{fig:diag1} (centre)
and Figure~\ref{fig:diag2} (left) are orientable, whereas the other diagrams of Figures~\ref{fig:diag1}, \ref{fig:diag2} are not. Then for $\beta=2$
\begin{equation}\label{eq:formula2} \widetilde{\mathcal{M}}(\bar{n}, \bar{N}, \bar{\tau})
= (1+o(1)) M^{k/3} \psi_2^\#(\bar{\alpha},\bar{s},\bar{t})~,\end{equation}
where
\begin{equation}\label{eq:defpsi2sh} \psi_2^\#(\bar\alpha, \bar{s}, \bar{t}) = \sum_{\text{orientable $\mathcal{D}$}} I^{\mathcal{D}}(\bar{\alpha},\bar{s},\bar{t})~.\end{equation}
\end{proof}

We conclude this section with a few remarks. First, one can altogether avoid moments and 
the Laplace transform. The modified moments
$\widetilde{\mathcal{M}}(\bar{n}, \bar{N}, \bar{\tau})$ are an approximation to the transform
\begin{equation*}  \mathbb{E} \prod_{p=1}^k \left[ M^{1/3} \sum_{j} \left\{ \frac{\sin ({\alpha_p} \sqrt{- \lambda_j(s_p,t_p)})}{ \sqrt{- \lambda_j(s_p,t_p)}} 
+ (-1)^{n_p} \frac{\sin ({\alpha_p} \sqrt{- \lambda_j^{\!\sim}(s_p,t_p)})}{ \sqrt{- \lambda_j^{\!\sim}(s_p,t_p)}} \right\} \right]~, 
\end{equation*}
where $\lambda_j^\sim$ are the rescaled eigenvalues of $-H$. 

Thus, the process
$\AD_\beta$ is characterised by its transform
\begin{equation}\label{eq:defpsibeta}
\psi_\beta(\bar{\alpha}, \bar{s}, \bar{t}) =
\mathbb{E} \prod_{p=1}^k  
\sum_{j=1}^\infty  \frac{\sin ({\alpha_p} \sqrt{- \lambda_j(s_p,t_p)})}{ \sqrt{- \lambda_j(s_p,t_p)}}
\end{equation}
defined in terms of the functions $\psi_\beta^\#$ of (\ref{eq:defpsi1sh}),(\ref{eq:defpsi2sh})
via the relation
\begin{equation}\label{eq:sh}
\sum_{I \subset \{1, \cdots, k\}} \psi_\beta(\bar{\alpha}|_I, \bar{s}|_I, \bar{t}|_I)
\psi_\beta(\bar{\alpha}|_{I^c}, \bar{s}|_{I^c}, \bar{t}|_{I^c}) = \psi_\beta^\#(\bar{\alpha}, \bar{s}, \bar{t})~.
\end{equation}
However, due to oscillations, the justification of convergence requires some effort, and so does 
the proof that the transform above determines the random process uniquely (see \cite{band}, 
where this strategy was carried out for a family of point processes without parameter dependence). 
Therefore we opted to return from modified moments to usual ones before passing to the
limit.

We recall that the formul{\ae} (\ref{eq:phipsi}) allow to go back from $\psi_\beta^\#$ to the
Laplace transforms $\phi_\beta^\#$.

Finally, the $k$-diagrams from the proof
of Lemma~\ref{l:sprop1} are in one to one correspondence with homotopy distinct ways to glue
$k$ discs with a marked point on the boundary into a manifold. The parameter $s$ associated to a 
diagram is related to the Euler characteristic $\chi$ of the manifold by the formula $s + \chi = 2k$.
The manifold corresponding to the diagram of Figure~\ref{fig:diag1} (centre) corresponds to
a torus glued from a disc, whereas the diagram of Figure~\ref{fig:diag2} (left) corresponds
to the sphere glued from two discs. 

Related topological expansions of Airy processes (for the case $\beta = 2$ and a single 
random matrix) appear in the work of  Okounkov \cite{Ok} on random permutations.

\section{Continuity estimates}\label{s:conv}

In this section, we consider a time-dependent random
matrix $\ltau{H}^{(N)}$ satisfying A1), A2), A3$_\beta$),
and the unimodularity assumption (\ref{eq:unimod}).
We prove the following two lemmata.

\begin{lemma}\label{l:cont-s} Let $n_1 \leq n_2 \leq \cdots  \leq n_{2k} = O(M^{1/3})$, $M/2 \leq N \leq 2M$, $\tau_1, \tau_2 = O(M^{-1/3})$. Then
\begin{multline}\label{eq:cont-s} \mathbb{E} \prod_{p=1}^k \left[ \tr P_{n_p}^{(N)}( ^{(\tau_1)}H^{(N)}  ) - \tr P_{n_p}^{(N)}( ^{(\tau_2)}H^{(N)}) \right] 
\\ \leq (C |\tau_1 - \tau_2|^k \sqrt{n_1n_2\cdots n_{2k}} + o(1)) M^{2k/3}~. 
\end{multline}
\end{lemma}

Note that the $o(1)$ term tends to zero as $M \to \infty$; however, for fixed $M$,
it does not tend to zero with $\tau_1 - \tau_2$. For, say, $k=2$, an estimate without the $o(1)$ term 
is presumably true for the Dyson Brownian motion
(Example~1 of the proem) but definitely not for the resampled Wigner matrices
(Example~2), since it would imply the existence of a continuous modification 
before the limit.

\begin{lemma}\label{l:cont-t}
Let $n_1 \leq n_2 \leq \cdots  \leq n_{2k} = O(N^{1/3})$, $N_1, N_2 = M + O(M^{2/3})$. Then
\begin{multline}\label{eq:cont-t}\mathbb{E} \prod_{p=1}^{2k}
\left[ \tr P_{n_p}^{(N_1)}( ^{(\tau)}H^{(N_1)}  ) - \tr P_{n_p}^{(N_2)}( ^{(\tau)}H^{(N_2)}) \right]
\\ \leq C \frac{|N_1 - N_2|^k}{M^k} \sqrt{n_1n_2\cdots n_{2k}}  M^{2k/3}~. 
\end{multline}
\end{lemma}

The proofs of the two lemmata are based on a combinatorial representation of the left-hand
side of (\ref{l:cont-s}), (\ref{l:cont-t}) which is derived from (\ref{eq:cpaths}).

\begin{proof}[Proof of Lemma~\ref{l:cont-s}]
We shall prove the lemma for $k = 2$; the general case requires no major changes.
According to the Cauchy--Schwarz inequality, it is sufficient
to consider the case $n_1 = n_2 = n_3 = n_4=n$. According to (\ref{eq:cheb}),
the left-hand side of (\ref{eq:cont-s}) is equal to
\begin{equation}\label{eq:chebdiff}
 (N-2)^{-2n} \mathbb{E} \sum \prod_{p=1}^4 \left\{ \prod_{i=1}^n \,^{(\tau_1)}\!H(u_{p,i}, u_{p,i+1})
- \prod_{i=1}^n \, ^{(\tau_2)}\!H(u_{p,i}, u_{p,i+1})\right\}~, 
\end{equation}
where the sum is over quadruples of closed non-backtracking paths. The second
part of Lemma~\ref{l:fs} implies that the contribution of quadruples in which
at least one edge is traversed more than twice can be absorbed in the $o(1)$ term
on the right-hand side of (\ref{eq:cont-s}). The contribution to (\ref{eq:chebdiff})
of a quadruple in which every edge is traversed exactly twice is equal to its contribution
to
\[ \mathbb{E} \prod_{p=1}^4 \left[ \tr P_{n_p}^{(N)}( ^{(\tau_1)}H^{(N)}  )  \right]~. \]
times
\[ (2 \left( 1 - (1 - |\tau_1 - \tau_2|)\right)^2 O(n^2) = O(|\tau_1 - \tau_2|^2 n^2)~.\]
An application of the first part of Lemma~\ref{l:fs} concludes the proof.
 \end{proof}

\begin{proof}[Proof of Lemma~\ref{l:cont-t}]
As in the proof of Lemma~\ref{l:cont-s}, we focus on $k = 2$, and we may
assume that $n_1 = n_2 = n_3 = n_4 = n$.  Denote
\[ A = \tr P_{n_p}^{(N_1)}( ^{(\tau)}H^{(N_1)}  ) - \tr P_{n_p}^{(N_2)}( ^{(\tau)}H^{(N_2)})~; \]
then $A = A(1) + A(2)$, where
\[ A(1) = \left[1 -  \frac{(N_1-2)^{n/2}}{(N_2 - 2)^{n/2}} \right] \tr P_{n}^{(N_1)}( ^{(\tau)}H^{(N_1)}  ) \]
and
\[ A(2) = \frac{(N_1-2)^{n/2}}{(N_2 - 2)^{n/2}}   \tr P_{n}^{(N_1)}( ^{(\tau)}H^{(N_1)}  )  - \tr P_{n}^{(N_2)}( ^{(\tau)}H^{(N_2)})~.\]
Due to the identity 
\[ \mathbb{E} A^4 = \sum_{\iota=1}^4 \binom{4}{\iota} \mathbb{E} A(1)^\iota A(2)^{4 - \iota} \]
and the Cauchy--Schwarz inequality, it is sufficient to show that $\mathbb{E}A(1)^4$
and $\mathbb{E} A(2)^4$ are bounded by the right-hand side of (\ref{eq:cont-t}).

First consider $\mathbb{E} A(1)^4$:
\[\begin{split}
A(1)
&= \left\{ 1 - \left[ 1 + \frac{N_2 - N_1}{N_2 -  2}\right]^{\frac{n}{2}}\right\} 
\tr P_n^{(N_1)} \left( \ltau{H}^{(N_1)} \right) \\
&= O\left(\frac{n}{M}(N_2 - N_1)\right) \tr P_n^{(N_1)} \left( \ltau{H}^{(N_1)} \right)~,
\end{split}\]
hence by the first part of Lemma~\ref{l:fs}
\[\begin{split}
\mathbb{E} A(1)^4 
&\leq \frac{Cn^4}{M^4} (N_2 - N_1)^4 n^4 \\
&\leq \frac{C'n^2}{M^2} (N_2 - N_1)^2 M^{4/3} \frac{n^6 (N_2 - N_1)^2}{M^{10/3}} \\
&\leq \frac{C''n^2}{M^2} (N_2 - N_1)^2 M^{4/3} 
 \end{split}\]

To estimate $\mathbb{E} A(2)^4$, we infer from (\ref{eq:cheb}) that
\[ A(2) = \pm \frac{1}{(N_2-2)^{n/2}} \sum \prod_{i=1}^n \ltau{H}(u_{i}, u_{i+1})~, \]
where the sum is over closed non-backtracking paths without loops $u_0 u_1 \cdots u_n$
for which $u_i \in \{1, \cdots, N_1 \vee N_2\}$ ($1 \leq i \leq n$), and at least one $u_i$ is special,
i.e.\ does not lie in
$\{1, \cdots, N_1 \wedge N_2\}$. Raising to the fourth power and taking expectation,
we  have:
\[ \mathbb{E} A(2)^4 = \sum  \mathbb{E} \prod_{p=1}^{4} \left[ (N_2-2)^{-n_2/2} {\prod_{i=1}^{n}} {^{(\tau)}}\!H(u_{p, i}, u_{p, i+1}) \right]~, \]
where the sum is over even $4$-tuples of paths 
of closed non-backtracking paths without loops
\[ u_{1,0}u_{1,1}u_{1,2}\cdots u_{1,n_1}; u_{2,0}u_{2,1}u_{2,2}\cdots u_{2,n_2};  \cdots; u_{4,0}u_{4,1}u_{4,2}\cdots u_{4,n_4}\]
in which $u_{p,i} \in \{1, \cdots, N_1 \vee N_2\}$, and
every one of the four paths contains at least
one special vertex.

Let us divide such $4$-tuples of paths according to the 
number $\ell$ of special vertices. 
The total contribution of paths with a given $\ell$
to $\mathbb{E}A(2)^4$ is at most 
\[ (C n |N_1-N_2| / M)^\ell \]
times the value of
\begin{equation}\label{eq:comp}
 \mathbb{E} \left[ \tr P_{n}^{(N_1\vee N_2)}( ^{(\tau)}H^{(N_1\vee N_2)}  ) \right]^4~;
 \end{equation}
here $C n$ bounds the number of ways to choose the location of the special vertices on
the paths, and $C |N_1-N_2|/M$ bounds the fraction of all special vertices (among the
possible vertices in $\{1, \cdots, N_1 \vee N_2\}$).
This yields the desired estimate for the contribution
of $4$-tuples of paths with $\ell \geq 2$. 

The  $4$-tuples of paths with $\ell=1$ ought to have at least one vertex common to the $4$ paths.
The contribution of such paths does not exceed
\[ C n^2 |N_1 - N_2| / M^2 \]
times the value of (\ref{eq:comp}) (cf.\ the combinatorial estimates of \cite[Part~II]{FS}),
and is therefore even smaller.
\end{proof}

\section{Point processes and line ensembles}\label{s:top}

\paragraph{Point collections} Consider the space $P$ of (semi-bounded) point collections
\[ P = \{ \Lambda = (\lambda_1 \geq \lambda_2 \geq \lambda_3 \cdots \geq - \infty) \} \]
equipped with the topology given by
\[ U_{k,\epsilon}(\Lambda) = \left\{ \widetilde{\Lambda} \in P \, \mid \, \widetilde{\lambda}_j \in U_\epsilon(\lambda_j) \quad
(1 \leq j \leq k)\right\} \quad (k \geq 1, \epsilon > 0)~,\]
where
\begin{equation}\label{eq:neigh}
U_\epsilon(a) = \begin{cases}
\left\{ b \in \mathbb{R} \, \mid \, |b-a|<\epsilon \right\}~, &a > -\infty \\
\left\{ b \in \mathbb{R} \,  \mid \, b < - 1/\epsilon \right\} \cup \{-\infty\}~, & a = -\infty
\end{cases}~.
\end{equation}
A finite collection of points
$\lambda_1 \geq \cdots \geq \lambda_N$ can be viewed as an element of $P$ by
setting
\[ \lambda_{N+1} = \lambda_{N+2} = \cdots = -\infty~. \]

The first lemma is a continuity property of the Laplace transform.
\begin{lemma}\label{l:1} Let $_M\Lambda$ ($M = 1,2,\cdots$) and $\Lambda$ be point collections in $P$.
If $_M\Lambda \to \Lambda$ in $P$ as $M \to \infty$ and
$\sum_j \exp(_M\lambda_j) \leq C$, then for every $\alpha > 1$
\[ \sum_j \exp(\alpha\, _M\!\lambda_j) \to \sum_j \exp(\alpha \lambda_j)~, \quad M \to \infty~. \]
\end{lemma}

\begin{proof}
The assumption $\sum_j \exp(_M\lambda_j) \leq C$ implies
\[ \# \left\{ j \, \mid \, _M\lambda_j \geq -R \right\}
 \leq \sum_{_M\lambda_j \geq - R} \exp(_M\lambda_j+R) \leq Ce^R~.
\]
Hence
\[ _M\lambda_j < R \quad (j > Ce^R) \]
and
\[ \lambda_j \leq R \quad (j > Ce^R)~. \]
Next,
\[\begin{split}
\sum_{_M\lambda_j \leq -R} \exp(\alpha \, _M\!\lambda_j)
    &\leq \sum_{_M\lambda_j \leq -R} \exp(_M\lambda_j) \exp(-(\alpha-1)R) \\
    &\leq \sum_j \exp(_M\lambda_j) \exp(-(\alpha-1)R) \leq C e^{-(\alpha-1)R}~,
  \end{split}\]
and also
\[ \sum_{\lambda_j \leq -R} \exp(\alpha \lambda_j) \leq C e^{-(\alpha-1)R}~.\]
For any $\epsilon > 0$, one can choose $R>0$ so that $Ce^{-(\alpha-1)R}<\epsilon/3$,
and $M_0$ such that for $M \geq M_0$
\[ \left|e^{\alpha \,_M\!\lambda_j}- e^{\alpha \lambda_j}\right| \leq \frac{\epsilon}{3Ce^R} \quad (1 \leq j \leq Ce^R)~.\]
This concludes the proof.
\end{proof}

The next lemma asserts that the Laplace transform is one-to-one.
\begin{lemma}\label{l:2}
Let $\Lambda, \widetilde\Lambda \in P$ be such that
\[ \sum_j \exp(\alpha \lambda_j) = \sum_j \exp(\alpha \widetilde\lambda_j) < \infty  \quad \left(\alpha \in (1, \infty) \cap \mathbb{Q}\right)~. \]
Then $\Lambda = \widetilde\Lambda$.
\end{lemma}

\begin{proof}
The Laplace transforms of the finite measures
\begin{equation}\label{eq:multbyexp} \sum \exp(2 \lambda_j) \delta_{\lambda_j}~, \quad 
\sum \exp(2 \widetilde\lambda_j) \delta_{\widetilde\lambda_j}
\end{equation}
coincide on $(-1, 1) \cap \mathbb{Q}$,  hence on the entire strip $|\Re \alpha| < 1$,
particularly, on the imaginary axis. The inversion formula for
the Fourier transform of finite measures implies that the two measures (\ref{eq:multbyexp})
coincide, and therefore $\Lambda = \widetilde\Lambda$.
\end{proof}

By a standard sub-subsequence argument, Lemmata~\ref{l:1} and \ref{l:2} imply:
\begin{lemma}\label{l:3}
Let $_M\Lambda$ ($M = 1,2,\cdots$) be point collections in $P$ such that
\[ \sum_j \exp(_M\lambda_j) \leq C \]
and
\[ \sum_j \exp(\alpha \, _M\!\lambda_j) \to \phi(\alpha) < \infty \quad (\alpha \in (1, \infty) \cap \mathbb{Q})~. \]
Then there exists $\Lambda \in P$ such that $_M\Lambda \to \Lambda$ in $P$ and
\[ \sum_j \exp(\alpha \lambda_j) = \phi(\alpha) \quad (\alpha \in (1, \infty) \cap \mathbb{Q})~.\]
\end{lemma}

Finally, Lemma~\ref{l:3} implies
\begin{cor}\label{c:4}
For any $C > 0$,
\[ \left\{ \Lambda \in P \, \mid \, \sum_j \exp(\lambda_j) \leq C \right\}
\]
is compact in $P$.
\end{cor}

\begin{proof}
According to Lemma~\ref{l:3}, it is sufficient to show that  
$\sum_j \exp(\alpha \lambda_j) \leq C_\alpha$ for every $\lambda$
in the set. As in the proof of Lemma~\ref{l:1},
\[ \# \{\lambda_j \geq 0 \} \leq C~, \]
and also $\lambda_1 \leq \log C$. Therefore
\[\sum_j \exp(\alpha \lambda_j) \leq C^{1+\alpha} +  C~.\]
\end{proof}

\paragraph{Parameter-dependent point collections}
Let $\T$ be a separable metric space (e.g.\ $\mathbb{R}^2$ equipped with the $\ell_1$
metric). Consider the space $P[\T]$ of parameter-dependent collections
\[ P[\T]= 
\left\{ \Lambda = \big(\lambda_1(\t) \geq \lambda_2(\t) \geq \cdots \geq -\infty \quad (\t \in \T) \big) \right\}~. \]
where $\lambda_j$ are assumed to be continuous functions (in the sense of (\ref{eq:neigh})). It is equipped with the topology
given by
\[ U_{k,\epsilon,I} = \left\{ \widetilde{\Lambda} \in L_d \, \mid  \, \widetilde{\lambda}_j(\t) \in U_\epsilon(\lambda_j(\t)) \quad (1 \leq j \leq k~, \, \t \in I)\right\}~, \]
where $k \geq 1$, $\epsilon > 0$, and $I \Subset \T$ is a compact subset.

\begin{lemma}\label{l:5}
Let
\[ \Lambda = \big(\lambda_1(\t) \geq \lambda_2(\t) \geq \cdots \geq -\infty \quad (\t \in \T) \big) \]
be a sequence of functions (continuity is not a priori given), such that
\[ \sum_j \exp(\lambda_j(\t)) \]
is bounded on compact subsets of $\T$. If
\[ \phi(\alpha, \t) =  \sum_j \exp(\alpha \lambda_j(\t)) \]
is continuous for every $\alpha \in (1,\infty) \cap \mathbb{Q}$, then $\lambda_j$ is continuous for
every $j$ (hence $\lambda \in P[\T]$).
\end{lemma}
\begin{proof}
Suppose $\t_k \to \t$ but $\lambda_j(\t_k) \nrightarrow \lambda_j(\t)$ for some $j$. According to Corollary~\ref{c:4},
one can find a subsequence of $\{\lambda(\t_k)\}_k$ that converges in $P$ to a limit $\mu$ which is distinct from $\lambda(\t)$.
But then Lemma~1  and the continuity of $\phi(\alpha,\cdot)$ imply that
\[ \phi(\alpha,\t) = \sum_j \exp(\alpha \mu_j)~, \quad (\alpha \in (1,\infty) \cap \mathbb{Q})~, \]
in contradiction to Lemma~\ref{l:2}.
\end{proof}

Setting
\[ \omega_\delta(f; \t) = \sup \left\{|f(\t_1) - f(\t)| \, \mid \, \|\t_1-\t\| \leq \delta \right\}
\quad (0 < \delta \leq 1)~, \]
we obtain
\begin{cor}\label{c:6} Let $\omega_\delta \to 0$ as $\delta \to 0$.
Any set of $\lambda \in P[\T]$ for which
\[ \sum_j \exp(\lambda_j(\t)) \]
is uniformly bounded on compact subsets of $\mathbb{R}^d$, and 
\begin{multline*}
\omega_\delta\left(\sum_j \exp(\alpha \lambda_j (\bullet)); \t\right) \leq C_\alpha(\t) \omega_\delta \quad
(0 < \delta < 1, \,\,\, \alpha \in [1, \infty) \cap \mathbb{Q})
\end{multline*}
is pre-compact in $P[\T]$.
\end{cor}

\paragraph{Point processes}

In this paper, a (semi-bounded) point process is a Borel random variable taking values in $
P$. The space of Borel measures on $P$ is denoted $\mathcal{P}$. In the remainder of this 
section, we keep the dependence on the element $\omega$ of the underlying probability 
space $\Omega$ explicit in our notation.

The following lemma gives a sufficient condition for convergence:
\begin{lemma}\label{l:pp:conv}
Let $_M\Lambda(\omega) = (_M\lambda_1(\omega) \geq \, _M\lambda_2(\omega) \geq \cdots)$ ($M \geq 1$) be a sequence
of (semi-bounded) point processes. Suppose for every $k \geq 1$ and every 
\[ \bar{\alpha} = (\alpha_1, \cdots, \alpha_k) \in ([1, \infty) \cap \mathbb{Q})^k~, \]
one has
\[ \mathbb{E}_\omega \prod_{p=1}^k \sum_j \exp(\alpha_p\, _M\!\lambda_j(\omega)) \rightarrow \phi({\bar{\alpha}})~, \]
such that for any $\alpha \in [1, \infty) \cap \mathbb{Q}$ 
\[1, \phi{(\alpha)}, \phi{(\alpha, \alpha)}, \phi{(\alpha, \alpha, \alpha)}~, \cdots\]
is a determinate moment sequence. 
Then the sequence $(_M\lambda)_M$ converges in distribution to a point process
\[ \Lambda(\omega) = (\lambda_1(\omega) \geq \lambda_2(\omega) \geq \cdots) \]
such that 
\[ \mathbb{E}_\omega \prod_{p=1}^k \sum_j \exp(\alpha_p \, \lambda_j(\omega)) = \phi({\bar{\alpha}})~.\]
\end{lemma}

\begin{proof}
Consider the random variables 
\[ _M\!\varphi({\alpha}) = \sum_j \exp \left\{ \alpha \, _M\!\lambda_j\right\}~, 
\quad \alpha \in [1, \infty) \cap \mathbb{Q}~.\]
For any tuple $(\alpha_1, \cdots, \alpha_k)$, the moments of the random vector 
\[ (_M\!\varphi({\alpha_1}), \,_M\!\varphi({\alpha_2}), \,\cdots, _M\!\varphi({\alpha_k}))\]
converge to the moments of a random vector for which the moment problem is
determinate (cf.\ Petersen \cite{Pet}). Thus the random variables $_M\varphi({\alpha})$
jointly converge in distribution. According to Corollary~\ref{c:4} and Lemma~\ref{l:3}, 
this implies the claim.
\end{proof}

\paragraph{Parameter-dependent point processes}

A parameter-dependent (semi-bounded) point process is for us
a function from the space of parameters $\T$
to (semi-bounded) point processes on $\mathbb{R}$.
Lemma~\ref{l:pp:conv} has the following counterpart:

\begin{lemma}\label{l:ppp:conv}
Let 
\[ _M\Lambda(\t, \omega) = (_M\lambda_1(\t, \omega) \geq \, _M\lambda_2(\t, \omega) \geq \cdots) \quad (\t \in  \T, \, M \geq 1) \]
be a sequence
of (semi-bounded) parameter-dependent point processes. If, for every $k \geq 1$ and every 
\[ \bar{\alpha} = (\alpha_1, \cdots, \alpha_k) \in ([1, \infty) \cap \mathbb{Q})^k~, 
\quad
\bar{\t} =(\t_1, \cdots, \t_k) \in \T^k~,\]
one has
\[ \mathbb{E}_\omega \prod_{p=1}^k \sum_j \exp(\alpha_p\, _M\!\lambda_j(\t_p, \omega)) \rightarrow \phi(\bar{\alpha};\bar{\t})~, \]
so that for every $\alpha \in [1, \infty) \cap \mathbb{Q}$ and every $\t \in \T$ 
\[1, \phi(\alpha; \t), \phi(\alpha, \alpha;\t, \t), \phi(\alpha, \alpha, \alpha;\t,\t,\t)~, \cdots\]
is a determinate moment sequence, 
then $(_M\Lambda)_M$ converges in distribution to a parameter-dependent point process
\[ \Lambda(\t, \omega) = (\lambda_1(\t, \omega) \geq \lambda_2(\t, \omega) \geq \cdots) \]
such that 
\[ \mathbb{E}_\omega \prod_{p=1}^k \sum_j \exp(\alpha_p \, \lambda_j(\t_p \omega)) = \phi(\bar{\alpha},\bar{\t})~.\]
\end{lemma}

\vspace{2mm}\noindent
A continuous modification of a parameter-dependent
point process (or just a continuous parameter-dependent point process) is a Borel random variable taking
values in $P[\T]$. The set of Borel probability measures
on $P[\T]$ is denoted $\mathcal{P}[\T]$; it is equipped with $w^*$ topology.  A sequence
\[ \big(_M\lambda(\t, \omega) \quad (\t \in \mathbb{R}^d, \omega \in \Omega)\big)_M \]
of continuous parameter-dependent point processes is called tight if the corresponding sequence of Borel measures in $\mathcal{P}[\T]$
is precompact.

The existence of a continuous modification can be verified with the help of
Kolmogorov's continuity lemma (cf.\ Billingsley \cite{Bil}):

\begin{lemma}\label{l:7} Let $_M\varphi(\t, \omega)$ ($\t \in \mathbb{R}^d$, $M = 1, 2, \cdots$) be a sequence of random functions defined on $\T =\mathbb{R}^d$, such that
\[ \mathbb{E}_\omega |_M\varphi(\t_1,\omega) - _M\!\varphi(\t, \omega)|^\eta \leq C \|\t_1-\t\|^{d+\varepsilon} \quad (\|\t - \t_1\|\leq 1)\]
for some $\eta,\varepsilon>0$. Then the random function $_M\varphi(\t, \omega)$ has a continuous 
modification for every $M$, and, moreover, the collection $(_M\varphi)_M$ is tight as a sequence
of random continuous functions.
\end{lemma}

We quote two corollaries of Lemma~\ref{l:7}:

\begin{lemma}\label{l:mod}
Let $\Lambda(\t, \omega)$ ($\t \in \mathbb{R}^d$) be a parameter-dependent point process. If
\begin{equation}\label{eq:condmod} \mathbb{E} \left(\sum_j \exp(\alpha \lambda_j(\t, \omega)) -  \sum_j \exp(\alpha \lambda_j(\t_1, \omega)) \right)^\eta
\leq C(\alpha, \t) \|\t - \t_1\|^{d + \varepsilon}\end{equation}
for every  $\alpha \in [1, \infty) \cap \mathbb{Q}$, $\|\t- \t_1\|\leq 1$, 
then $\Lambda$ has a continuous modification.
\end{lemma}

\begin{lemma}\label{l:prop} Let $(_M\Lambda(\t, \omega))_M$ ($\t \in \mathbb{R}^d$) be a 
sequence of continuous parameter-dependent point processes, such that the Laplace 
transforms
\[ _M\varphi(\alpha; \t;  \omega) = \sum_j \exp(\alpha \,\,_M\!\lambda_j(\t, \omega))~, \quad
(\alpha \in [1, \infty) \cap \mathbb{Q}) \]
satisfy
\[ \mathbb{E}_\omega {_M}\!\varphi(1; \t_0, \omega) \leq C(\t) \quad (\t \in \mathbb{R}^d)\]
and
\[ \mathbb{E}_\omega |{_M}\!{\varphi}(\alpha; \t; \omega) - _M\!\varphi(\alpha;\t_1;\omega)|^\eta
\leq C_\alpha \|\t-\t_1\|^{d+\varepsilon} \quad (\|\t-\t_1\|\leq 1~, \,\alpha \in (1, \infty) \cap \mathbb{Q})~. \]
Then the sequence $(_M\Lambda(\t, \omega))_M$ is tight.
\end{lemma}

\section{Proofs of the theorems}\label{s:pf}

\begin{proof}[Proof of Theorem~\ref{thm:sosh} and Theorem~\ref{thm:prop} (item 1)]
We closely follow the method of \cite{Sosh}. Let us verify that for any 
\[ \overline{\t} = (\t_1, \cdots, \t_k) \in \mathbb{R}^{2k}~,
\quad \overline{\alpha} = (\alpha_1, \cdots, \alpha_k) \in (0, \infty)^k \]
the transforms 
\[ \mathbb{E} \prod_{p=1}^k \sum_j \exp \left\{ \alpha_p \, _M\!\lambda_j(\t_p) \right\} \]
converge, as $ M \to \infty$, to a limit $\phi_\beta(\overline{\alpha}, \overline{\t})$ which 
is related to $\phi_\beta^\#$ of Lemma~\ref{l:sprop1} by the equations
\[ \phi_\beta^\#(\overline{\alpha}, \overline{\t})
= \sum_{I \subset \{1,\cdots,k\}} \phi_\beta(\overline{\alpha}|_I, \overline{\t}|_I)
\phi_\beta(\overline{\alpha}|_{I^c}, \overline{\t}|_{I^c})~.\]
The statements will then follow from Lemma~\ref{l:pp:conv} (the required determinacy
property of the processes at a fixed $\t$ follows from the determinantal / Pfaffian structure,
cf.\ \cite[Proof of Theorem~A]{Sosh}).

In view of Lemma~\ref{l:sprop1}, it is sufficient to show that for
\[ m_p \sim \alpha_p  M^{2/3}~,\quad  \tau_p = s_p M^{-1/3}~,
\quad N_P = M (1 + 2 t_p M^{-1/3})\]
one has:
\begin{equation}\label{eq:req} \mathcal{M}(\bar{m}, \bar{\tau}, \bar{N}) = 
\mathbb{E} \prod_{p=1}^k \sum_j \left\{  e^{\alpha_p \, _M\!\lambda_j(s_p, t_p)}+ 
(-1)^{m_p} e^{\alpha_p \, _M\!\lambda_j^\sim(s_p, t_p)}\right\} 
+ o(1)~,\end{equation}
where $_M\!\lambda_j^\sim$ are obtained by scaling (\ref{eq:scaling}) the eigenvalues of
$\xi_j^\sim = \xi_{M-j+1}$ of $-\ltau{H}$ (the estimates need not be uniform in $s_p$ and $t_p$).

To prove (\ref{eq:req}), it is sufficient to divide the  eigenvalues into four categories:
(a) those near the right edge: $|\xi_j - 2\sqrt{N}| \leq M^{-1/6 + \epsilon}$ for, say, $\epsilon = 1/100$; 
(b) the left edge, $|\xi_j  + 2\sqrt{N}| \leq M^{-1/6 + \epsilon}$; (c) the bulk, $|\xi_j| \leq 2 \sqrt{N} - M^{-1/6+\epsilon}$; and
(d) the tails, $|\xi_j| \geq 2 \sqrt{N} +   N^{-1/6+\epsilon}$. The first two categories yield the two terms of
(\ref{eq:req}); the third one vanishes in the limit since 
\[ M (1 - M^{-2/3+\epsilon})^{\alpha M^{2/3}} \to 0~, \]
whereas the fourth one vanishes in the limit due to Lemma~\ref{c:sosh} and Chebyshev's inequality.
\end{proof}

\begin{proof}[Proof of Theorem~\ref{thm:prop} (item 2)]
Let us verify the condition (\ref{eq:condmod}) of Lemma~\ref{l:mod} with $d=2$, 
$\eta = 6$, and $\varepsilon = 1$. Consider a matrix 
$\ltau{H}$ with unimodular entries (\ref{eq:unimod}); in view of Theorem~\ref{thm:sosh}, it is 
sufficient to prove that (\ref{eq:condmod}) holds before the limit, up to an $o(1)$ term which
may depend on $\t,\t_1 \in \mathbb{R}^2$ but vanishes as $M \to \infty$. The approximation
(\ref{eq:req}) justified in the proof of Theorem~\ref{thm:sosh} reduces our task to an estimate 
for moments. Finally, the argument from the proof of Lemma~\ref{l:sprop} reduces it further 
to the estimates on the  modified moments 
$\widetilde{\mathcal{M}}$, which are provided by Lemmata~\ref{l:cont-s} and \ref{l:cont-t}
(applied with $k=3$).
\end{proof}

\begin{proof}[Proof of Theorem~\ref{thm:conv}]
The proof is parallel to that of Theorem~\ref{thm:prop} (item 2): we verify the conditions
of Lemma~\ref{l:prop} (with $d=1$, $\eta = 4$, and $\varepsilon = 1$) using Lemma~\ref{l:cont-t} 
(with $k=2$); the only significant difference is that now we can not tolerate any $o(1)$ terms. Luckily, 
there is no such term in Lemma~\ref{l:cont-t}, therefore  we only need to deal with the $o(1)$ 
term of (\ref{eq:req}). 

This can be done as follows. For $|\lambda| \leq M^{\epsilon}$,  the approximation of 
$e^{\alpha \lambda}$ by a power of $1  +\lambda / (2N^{2/3})$ can be differentiated in $\lambda$,
hence the contribution of these eigenvalues to additive $o(1)$ term can be incorporated in 
a multiplicative term $1 + o(1)$. The exponentially small contribution of the bulk and the tails 
can be incorporated in $|t_1 - t_2|^2$ since, by construction, we  only need to consider 
$|t_1 - t_2| \gtrsim M^{-2/3}$.
\end{proof}

\end{document}